\documentclass[a4paper,UKenglish,cleveref,thm-restate,pdfa]{lipics-v2021}

\newcommand{\parag}[1]{\subparagraph{#1}}
\newcommand{\itemname}[1]{\textbf{\textsf{#1}}}

\newcommand{\sus}{\mathsf{SUS}}
\newcommand{\lcp}{\mathsf{LCP}}
\newcommand{\last}{\mathsf{last}}
\newcommand{\lt}{<}

\usepackage{ragged2e} 
\usepackage{graphicx}
\usepackage{svg}
\usepackage{balance}
\usepackage{xcolor}
\usepackage{xspace}
\usepackage[utf8]{inputenc}
\usepackage{float}
\usepackage{amssymb}
\usepackage{amsmath}
\usepackage{multirow}
\usepackage{booktabs}
\usepackage{makecell}
\usepackage{amsthm}
\usepackage{thmtools}
\usepackage[makeroom]{cancel}
\usepackage{xcolor, soul}
\usepackage{siunitx}
\usepackage{tablefootnote}

\usepackage{fancyvrb}
\VerbatimFootnotes

\usepackage[noEnd=true,indLines=true]{algpseudocodex}
\usepackage[ruled]{algorithm}
\tikzset{algpxIndentLine/.style={draw=black}}

\usepackage{afterpage}

\theoremstyle{theorem}
\newtheorem{mytheorem}[theorem]{Theorem}

\newtheorem{newdefinition}[theorem]{Definition}

\title{The Anti-Lexicographic SUS-Anchor:\\ An Empirically Optimal Selection Scheme}
\titlerunning{The Anti-Lexicographic SUS-Anchor}

\author{Ragnar Groot Koerkamp}{Karlsruhe Institute of Technology,
  Germany}{ragnar.grootkoerkamp@gmail.com}{https://orcid.org/0000-0002-2091-1237}{}


\authorrunning{R. Groot Koerkamp}
\Copyright{Ragnar Groot Koerkamp}

\supplement{}
\supplementdetails[subcategory={Source Code}, cite=minimizers-repo,
  swhid={swh:1:dir:6b775d1315385c94ecbd6058cfbba2476fbeb0b4}, swhdelimiter={,}]{Software}{https://github.com/RagnarGrootKoerkamp/minimizers}

\acknowledgements{I thank the anonymous reviewers for their helpful feedback.}


\nolinenumbers{}
\EventEditors{Nadia El-Mabrouk and Fabio Vandin}
\EventNoEds{2}
\EventLongTitle{26th International Conference on Algorithms for Bioinformatics (WABI 2026)}
\EventShortTitle{WABI 2026}
\EventAcronym{WABI}
\EventYear{2026}
\EventDate{August 31--September 2, 2026}
\EventLocation{L'Aquila, Italy}
\EventLogo{}
\SeriesVolume{390}
\ArticleNo{22}

\begin{document}

\maketitle
\begin{abstract}
\vspace{-2\baselineskip}
\subparagraph{Motivation.}
\emph{Selection schemes} provide a way to select a subset of positions in a text in
such a way that no two consecutive selected positions are more than \(w\) apart.
These selected positions can be used as ``anchor'' points for text indices
such that every sufficiently long pattern corresponds to at least one anchor
\cite{lc-anchors}.
Closely related are \emph{sampling schemes}, that sample a \emph{\(k\)-mer} from each window
of \(w\) consecutive \(k\)-mers in a text, and the more restricted \emph{minimizer
schemes}, that achieve this by taking the smallest \(k\)-mer according to some
order.  In recent years, there has been a renewed interest in the search for
\emph{low density} schemes that select/sample only a small fraction of
positions/\(k\)-mers.

The mod-minimizer \cite{modmini} provides a near-optimal density of \(1/w\) as
\(k / w \to \infty\), while schemes such as the greedy minimizer work well for explicit
small parameters roughly in the regime \(k \leq 2w\), for \(k\) and \(w\) up
to \(15\) or so.

When \(k < \log_\sigma w\) is small, minimizer schemes cannot do well
\cite{asymptotic-optimal-minimizers}. 
As a first step towards low density sampling schemes in this regime, we fix
\(k=1\) and search for a near-optimal selection scheme to improve the existing
bidirectional string anchors (bd-anchors) \cite{bdanchors,lc-anchors}.

\vspace{-0.5em}
\subparagraph{Methods.}
Inspired by bd-anchors, we introduce the \emph{smallest unique substring} or SUS-anchor:
given a window, this considers all suffixes that do not occur as a substring
elsewhere in the window. It then samples the start position of the smallest suffix according to the
new \emph{anti-lexicographic} order that minimizes the first character and maximizes
the remaining characters.
We give a linear-time and \(O(w)\) space streaming algorithm to compute all
SUS-anchors of a string.

\vspace{-0.5em}
\subparagraph{Results.}
For alphabet size \(\sigma=4\) and \(k=1\), the parameter-free anti-lexicographic
SUS-anchor empirically has density \(<1\%\) away from the density lower bound and
is at least \(4\times\) closer to the lower bound than all other tested schemes.
For alphabet size \(\sigma=2\), the density is at most \(10\%\) above the lower
bound, which still improves \(2\times\) to \(3\times\) the overhead of the random
minimizer and is consistently better than the greedy minimizer.
Likewise, the anti-lexicographic minimizer performs better than all other schemes
apart from the greedy minimizer.
\end{abstract}

\section{Introduction}
\label{sec:introduction}
Minimizers \cite{winnowing,minimizers} are a technique to subsample the
\emph{\(k\)-mers} of a string such that consecutive samples are at most \(w\) positions apart.
Minimizer sampling has many applications in bioinformatics.
In particular, a number of data structures exploit this \emph{window guarantee} and
use minimizers as a locality-sensitive hash: SSHash
\cite{sshash,sshash2} and \(k\)ache-hash \cite{kache-hash-preprint}
use the minimizer of a string as the key in a hash map, while the U-index
\cite{u-index} is a text index that considers the minimizers of a pattern.
In a similar way, the text index of Loukides et al. and Ayad et al.
\cite{bdanchors,lc-anchors} works for patterns of length at least \(\ell\) and selects \emph{positions} in the text (\emph{anchors}) that are at most \(w=\ell\)
apart in a locally consistent way, so that each queried pattern of length \(\geq
\ell\) can be found via its anchor. Improving their bidirectional string anchors
is the main goal of this work, as this directly leads to a smaller text index.

From another point of view, the minimizers of a string are a \emph{sketch} or
compressed~(lossy) representation, allowing for faster processing.  Specific
applications include seeding for read-mapping as done by minimap
\cite{minimap}, the minimizer-space De Bruijn graph \cite{mdbg}, sketching for
Jaccard similarity in mashmap \cite{mashmap}, host depletion in Deacon
\cite{deacon-preprint}, and more \cite{minimizer-sketches}.
In this direction, there are also other schemes such as syncmers
\cite{syncmers} and universal hitting sets~\cite{docks-wabi,docks,improved-minimizers}. In these schemes, each \(k\)-mer
is sampled independently from its context, which may loosen the window guarantee
and does not directly give a suitable key or anchor for each window.

\parag{Recent work.}
There is a long line of active work on minimizers.
Of particular interest is the \emph{density}, which is the expected fraction of
sampled positions.
Recently, the mod-minimizer~\cite{modmini,oc-modmini} introduced schemes with near-optimal
density for
\(k\to\infty\). The GreedyMini~\cite{greedymini} is the state of the art
for smaller \(w\) and \(k\) roughly up to \(15\), achieving particularly close to
optimal density for \(k=w+1\), \(k=w\), and \(k\) just below \(w\).
Shur et al. recently introduced
\emph{spacers} \cite{10-minimizers} that prefer sampling \(k\)-mers starting with
\texttt{10} (after projecting to a binary alphabet) for which the next occurrence of \texttt{10} is as far ahead as possible.
Spacers improve the density of the ABB+ scheme we introduce here, and are
especially good when \(w\) is large.
For particularly small parameters \((\sigma,k)=(2,\leq 6)\) and \((\sigma,k)=(4, 2)\),
exactly optimal \emph{minimizer} schemes are given by the search algorithm OptMini \cite{optmini} and
analysed theoretically by Shur \cite{on-minimizers-of-minimum-density}.

This raises the question: \textbf{can we find near-optimal \emph{sampling} schemes for small \(k\), and in
particular for \(k=1\)}?
An ILP search suggests that the answer is yes: this is always possible when
\(k\equiv 1\mod w\)~\cite{sampling-lower-bound}, and thus we are motivated to search for a ``clean'',
constant-space scheme.

Some further recent work includes SimdMinimizers \cite{simd-minimizers},
a SIMD-based algorithm that computes \emph{random minimizers} at around 500 Mbp/s.
A precise analysis of the density of the random minimizer is given in \cite{random-mini-density}.
Vigemers \cite{vigemers} reduce the imbalance when using minimizers for
partitioning, and multiminimizers \cite{multiminimizers} reduce the density at
the cost of more compute.

\parag{Sampling and selection schemes.}
Minimizer schemes hash each \(k\)-mer and sample the \(k\)-mer with the smallest
hash in a window of \(w\) \(k\)-mers (or \(\ell=w+k-1\) characters).
Minimizers cannot achieve a good density of \(O(1/w)\) for constant \(k\) as
\(w\to\infty\) \cite{asymptotic-optimal-minimizers},
but this restriction does not apply to \emph{sampling schemes}, which have more
freedom because they are not required to first hash all \(k\)-mers.
In particular, \emph{selection schemes} with \(k=1\) simply sample a single \emph{position}
based on a window of length \(w=\ell\).
The main reason that selection schemes can achieve better density than
sampling schemes is that for a fixed sized window, they can also sample the last
\(k-1\) positions, and thus achieve a greater distance between consecutive samples.
Bidirectional string anchors (bd-anchors) are one such selection scheme~\cite{bdanchors}. These sample the start position of the
smallest \emph{rotation} of a window and achieve a density of \(O(1/w)\).
For large \(w\) and alphabet size \(\sigma=4\), bd-anchors have a density around
\(20\%\) above the lower bound. Our goal is to reduce this to as close to \(0\%\)
overhead as possible.

\enlargethispage{3em}
\parag{Contributions.}
In this paper, we introduce SUS-anchors, short for \emph{smallest unique substring} anchors.
As the name implies, given a window, these sample the start position of the
smallest substring that does not occur a second time. The \emph{lexicographic}
SUS-anchor is equivalent (after reversing the order of the alphabet) to the
the \emph{maximal suffix} of each window. 
Our main results are as follows:
\begin{itemize}
\item We experimentally show that SUS-anchors with the anti-lexicographic order (see
below) have density within
\(1\%\) of optimal for \(\sigma=4\) and any \(w\).
\item We give an \(O(n)\) time and \(O(\ell)\) space sliding-window algorithm to compute the
SUS-anchors of a string of length \(n\) that is based on the \emph{monotone-queue}
approach \cite{simd-minimizers}.
\end{itemize}

\noindent
Smaller contributions that are of possible independent interest are:
\begin{itemize}
\item \emph{character-based orders} that generalize e.g. the
lexicographic, alternating \cite{minimizers}, and ABB~\cite{minimally-overlapping-words} order; and
\item two new character-based orders: the \emph{ABB+} and \emph{anti-lexicographic} order.
\end{itemize}

\noindent
Parts of these results have been previously introduced in the
author's PhD thesis \cite{thesis}.
\section{Preliminaries}
\label{sec:preliminaries}
\subsection{Minimizer, sampling, and selection schemes}
\label{sec:minimizer-sampling-and-selection-schemes}
For a more in-depth introduction to minimizers, we refer the reader
to the survey by Zheng et al. \cite{minimizer-sketches},
the mod-minimizer paper
\cite{modmini}, and the author's thesis \cite{thesis}. Here we briefly
introduce the required notation and concepts.

\parag{Notation.}
For an integer \(w\), we write \([w] = \{0, 1, \dots, w-1\}\), and for a string
\(W=w_0\dots w_{|W|-1}\)
we use \(W[i\dots j)\) to indicate the substring \(w_i\dots w_{j-1}\).
We assume an alphabet \(\Sigma\) of size \(\sigma\).

\parag{Sampling schemes.}
Given parameters \(w\geq 1\) and \(k\geq 1\), a \emph{window} is a string over \(\Sigma\) of length
\(\ell=w+k-1\) that contains exactly \(w\) \(k\)-mers.
A \emph{local sampling scheme} or just \emph{sampling scheme} is a function \(f:
\Sigma^\ell \to [w]\) that indicates that from window \(W\), the \(k\)-mer
\(W[f(W)\dots f(W)+k)\) is \emph{sampled}.

Given a text \(T\), we are interested in the \emph{set} of all sampled positions when
sliding window \(W\) over the text.
Consecutive sampled positions differ by at most \(w\), and
a sampling scheme is \emph{forward} when the
sampled position in \(T\) never decreases when sliding the window forward.

Often considered are \emph{minimizer schemes} \cite{winnowing,minimizers}, which are forward sampling schemes that sample (the 
position of) the smallest \(k\)-mer in the window according to some order as given by a hash function.

\parag{Selection schemes.}
In this paper, we are particularly interested in \emph{selection} schemes, which are
sampling schemes with \(k=1\) \cite{small-uhs}. In particular, these select a
\emph{position} rather than \emph{\(k\)-mer} in each window. Given a fixed window
(context) of \(\ell\) characters, the \(w=\ell\) selection scheme has greater
freedom than an \(w=\ell-k+1\) sampling scheme, and thus can achieve lower density.

\parag{Density.}
The \emph{particular density} of a sampling scheme is the fraction of positions that
are sampled. The \emph{density} is the expected value of the particular density on a
random string with length going to infinity.
For forward schemes, the density can be computed as the probability that
two \emph{consecutive} windows \(W\) and \(W'\) (overlapping by \(\ell-1\) characters and
together forming a \emph{context} of \(\ell+1\) characters) sample a different position.

\parag{Density lower bounds.}
A trivial lower bound on the density of sampling schemes is \(1/w\), since at
least one in every \(w\) positions is sampled. The best
current lower bound for \emph{forward} sampling schemes \cite{sampling-lower-bound}
simplifies (with a small loss of
accuracy) to \(\frac 1{w+k} \left\lceil \frac{w+k}w\right\rceil\).
For \(k=1\), this gives a lower bound for forward selection schemes of
\(\frac1{w+1}\left\lceil\frac{w+1}w\right\rceil = \frac 2{w+1}\).
The main insight is that the density equals the expected number of samples
in a random cyclic string (cycle) of length \(w+1\), and in any such cycle,
at least two different positions must be sampled.

For \emph{minimizer} schemes, a lower bound is given by \(1/\sigma^k\) \cite{asymptotic-optimal-minimizers}: if
\(w\to\infty\), exactly all occurrences of the smallest \(k\)-mer will be sampled,
and these occur every \(\sigma^k\) positions in expectation.
However, this bound does \emph{not} hold for general sampling schemes, and
it is exactly this property that will allow us to design near-optimal sampling
schemes for small \(k\).
\subsection{Bidirectional string anchors}
\label{bd-anchors}
\emph{Bidirectional string anchors} (\emph{bd-anchors}) are a \(k=1\) selection scheme introduced by
Loukides, Pissis, and Sweering for the purpose of
text indexing \cite{bdanchors}.
Given a window, they sample the (leftmost) start position of the
lexicographically smallest rotation.
A drawback of bd-anchors is that they are not forward: the window \texttt{ZABAAC} has
\texttt{AAC...} as smallest rotation, while the shifted window \texttt{ABAACA} has \texttt{AAB...} as
smallest rotation. After further shifting to \texttt{BAACAY}, the smallest rotation is \texttt{AAC...}
again: the \texttt{AB} prefix caused the selected position to jump around.
To ensure the density\footnote{Even though \(\ell=w+k-1=w\), we
will use \(w\) for the density and \(\ell\) for the length of the window.} is in \(O(1/w)\),
\emph{reduced} bd-anchors \cite{bdanchors} avoid sampling the last \(r=C\lceil \log_\sigma \ell\rceil\)
positions, where \(C=4\) in theory but in practice \(C=3\) or less
suffices.
Given the choice of \(r>\log_\sigma \ell\), in most windows of a random string
the bd-anchor is the same as the smallest lexicographic \((r+1)\)-mer.
Note that reduced bd-anchors are still not forward.

In practice, bd-anchors can be computed in \(O(n\ell)\) time by using Booth's linear-time
algorithm for the lexicographically minimal rotation \cite{smallest-rotation}.
An \(O(n)\) time algorithm is possible using a data
structure of Kociumaka \cite{minimal-suffix-cpm16},
but this is mostly of theoretical interest only since it requires \(O(n)\)
words of space for e.g. a suffix array. 
Theorem 3 of \cite{bdanchors} gives an \(O(n)\) time and \(O(\ell)\) space
algorithm when the alphabet size is constant (\(\sigma = O(1)\)) by using Kociumaka's method on overlapping chunks of size e.g. \(2\ell\).
A practical linear time algorithm to compute reduced bd-anchors with \(r=\lceil
4\log_\sigma w\rceil\) is given by \cite{lc-anchors}. This works by first
computing the lexicographic minimizers of size \(k=r\) and only comparing
the corresponding rotations in case of ties. Additionally, the \emph{randomized}
reduced bd-anchor is introduced that mimics the random minimizer: it considers
the smallest \(r\)-mer in a window under a pseudo-random order and breaks ties
by preferring the lexicographically smallest rotation. In practice, the density
of randomized reduced bd-anchors is indistinguishable from the density of random minimizers.

\enlargethispage{2em}
\subsection{Maximal suffixes}
\label{sec:maximal-suffixes}
Separate from the literature on minimizers, there is a line of work on
the \emph{non-empty minimal suffix} and \emph{maximal suffix} of a string \cite{factorizing-words}.
While these two problems appear similar, as one might simply reverse the order
of the alphabet, a crucial point is that a string \(A\) is always smaller than
\(B\) when \(A\) is a prefix of \(B\).
Thus, the minimal suffix prefers shorter suffixes, while the maximal suffix
prefers longer suffixes. Because of this, minimal suffixes are less stable as we
slide a window over a text, and we do not further consider them.

Babenko et al.~\cite{minimal-maximal-suffix-cpm13} introduce an algorithm that preprocesses a
string of length \(n\) into a linear-space data structure and can then find the
maximal suffix of query substrings of length \(\ell\) in \(O(\log \ell)\) time, which 
was improved to \(O(1)\) query time in \cite{minimal-maximal-suffix-tcs}.
\section{Character-based orders}
\label{sec:character-based-orders}
To capture some of the existing ``lexicographic-like'' orders, we define
\emph{character-based orders} that compare strings one character at a time\footnote{A slight generalization of this concept that we do not otherwise need in this paper uses orders \(O_i\) on strings of length \(i\) and then compares \emph{prefixes} \(a_0\dots a_\ell \lt_{O_\ell} b_0\dots b_\ell\) instead.}.
This is a natural requirement in our setting, because it allows comparing
suffixes of a string without worrying that future characters (after shifting the
window) will change the order, as long as one is not a prefix of the other.

\begin{newdefinition}[Character-based order]
An order \(O\) on strings over \(\Sigma\) is \emph{character-based} if there exist
orders \(O_i\) on \(\Sigma\) for \(i\in \{0, 1, 2, \dots\}\) such that for
all strings \(A=a_0\dots a_{|A|-1}\) and \(B=b_0\dots b_{|B|-1}\) with longest common
prefix \(\ell = \mathsf{lcp}(A,B)\) we have
\begin{equation*}
A <_O B \quad\quad \text{iff} \quad\quad (\text{$A$ is a strict prefix of $B$}) \text{ or } a_\ell <_{O_\ell} b_{\ell}.
\end{equation*}

The orders can be either \emph{total orders}, or \emph{linear preorders} when equalities
are allowed.

\end{newdefinition}

This scheme encapsulates the \textbf{lexicographic order} on strings, where each \(O_i\) is
simply the lexicographic order on \(\Sigma\).
The main drawback of the lexicographic order is that it clusters small strings:
since \texttt{AAAAX} is small, the next \(k\)-mer \texttt{AAAXY} is also small, possibly causing
consecutive positions to be sampled as minimizers.
Most of the following schemes instead look for a \emph{transition} from a small
character (\texttt{A}) to a large (\texttt{Z}) or non-small (\texttt{BCD..Z}) character.

The \textbf{alternating order} \cite{minimizers} uses the lexicographic order for even \(i\), and the reverse
lexicographic order for odd \(i\), so that \texttt{AZAZAZ...} is the smallest string. The \textbf{ABB order} \cite{non-overlapping-codes,minimally-overlapping-words} uses the lexicographic order for
\(O_0\), and for \(i>0\) it uses the order
\[1=_{O_i}2=_{O_i}\dots=_{O_i}\sigma-1 <_{O_i} 0,\]
so that any string like \texttt{ABBBB...} or \texttt{AXYZDEF...} is minimal.
This scheme has the nice property that occurrences of small strings starting in
\texttt{A} and not containing further \texttt{A}'s are disjoint \cite{non-overlapping-codes}.
(The concept of sets of minimally overlapping strings has been
further explored in \cite{minimally-overlapping-words}, but these do not form a
character-based order and do not directly work with variable-length strings.)
\textbf{Vigemers} \cite{vigemers} are also a character-based order, where \(O_i\) is the
order after xor'ing by a character \(\gamma_i\).

A drawback of the ABB order is that it throws away some information: for example,
over the normal alphabet, \texttt{AB} and \texttt{AC} are considered equal. Thus, we also consider a version with tiebreaking, \emph{ABB+}:
\begin{newdefinition}[ABB+ order]
The \emph{ABB+} order first compares two strings via the ABB order, and, in case of
a tie, compares them via the plain lexicographic order.

\end{newdefinition}
Note that this is only useful when comparing strings (\(k\)-mers) of equal
length, and that the ABB+ order is not itself a character-based order.

The following scheme is more practical.
\begin{newdefinition}[Anti-lexicographic order]
The \emph{anti-lexicographic order} uses the lexicographic order for \(O_0\), and
\emph{reverse} lexicographic order for \(O_i\) for \(i>0\).

\end{newdefinition}
In this order, the smallest string is \texttt{AZZZZ...}.
\section{Smallest unique substring anchors}
\label{sec:smallest-unique-substring-anchors}
\parag{Intuition.}
Consider again the bd-anchor, which samples the start position of the smallest
rotation of a window \(W\).
As we saw in \cref{bd-anchors}, a drawback is that rotations ``wrap around'', and that
the character \(W[0]\) at
the start of a string influences whether the rotation starting at the last
character \(W[\ell-1]\) is small or not. This is easily fixed by considering only the smallest
\emph{suffix} instead. However, the smallest suffix is the empty suffix, and so we
could sample the first character of the smallest \emph{non-empty} suffix.
Unfortunately, this results in a bad density: a random window ends
in the smallest symbol \(0\) with probability \(1/\sigma\), in which case the suffix
consisting of just this character is the smallest one, and the density will be
around \(1/\sigma\) regardless of \(w\).
To avoid this, we impose the following restriction (which, in a way,
generalizes the ``non-empty'' condition): a suffix is only allowed to
be sampled if it does not occur elsewhere in the window as a substring.
Thus, we look for the \emph{smallest unique suffix} \(W[i\dots)\).
Let \(S=W[i\dots j)\) be the shortest prefix of \(W[i\dots)\) that is unique in~\(W\).
Then \(S\) is the \emph{smallest unique substring} (SUS\footnote{Not to be confused with the \emph{shortest} unique substring, which is
also commonly abbreviated as SUS.}) of \(W\) \cite{thesis}.
As an example, the string \texttt{CABBAB} has \texttt{AB} as its smallest suffix, and \texttt{ABBAB}
as its smallest \emph{unique} suffix. The smallest unique \emph{substring} is \texttt{ABB}, and
the index of the correspondingly sampled \emph{SUS-anchor} is 1.

\begin{newdefinition}[SUS-anchor]
Given a window \(W\) of length \(w=\ell\), the \emph{smallest unique substring}
is the smallest substring \(\sus(W)=W[i\dots j)\) that does not occur elsewhere in \(W\).
Then \(W[i\dots \ell)\) is the smallest unique \emph{suffix}, and the SUS-anchor is \(i\).

\end{newdefinition}

We note that for the purposes of the definition, it would be sufficient to
only consider the starting position of smallest unique \emph{suffix}.
Nevertheless, we prefer to consider the smallest unique \emph{substring} since it has properties that are similar to
the \(k\)-mers sampled by minimizer schemes. For example, both
optimal-density minimizer schemes and optimal-density SUS-anchors would always
sample non-overlapping \(k\)-mers and SUS-anchors respectively.

\parag{Variants.} Alongside the lexicographic variant,
the SUS-anchor allows variants based on character-based orders.
Specifically, we consider SUS-anchors with the anti-lexicographic order.
These variants work just like the lexicographic version: simply consider the set of
suffixes that do not occur as a substring elsewhere, and then take the
smallest of these using the chosen character-based order.

\parag{MS-anchor.}
It turns out that the concept of \emph{lexicographic smallest unique suffix} is exactly equivalent
to that of the \emph{maximal suffix}
\cite{factorizing-words,minimal-maximal-suffix-tcs} after reversing the order of the alphabet:
in both cases, we look for the extremal suffix where longer suffixes should be
preferred over shorter ones.
More generally, for any character-based order \(O\), the SUS-anchor under \(O\)
starts at the same position as the maximal suffix when suffixes are compared by
the reverse of \(O\).
Whereas the SUS-anchor explicitly requires considering a unique suffix, this is
handled implicitly by the maximal suffix: when one suffix is a prefix of a
longer suffix, the longer suffix is automatically the larger one (regardless of
the chosen character-based order \(O\)).

\begin{newdefinition}[MS-anchor]
Given a window \(W\) of length \(w=\ell\),
the \emph{maximal suffix anchor} or MS-anchor samples the position \(i\in[w]\) where the maximal
suffix \(W[i\dots \ell)\) of \(W\) starts.

\end{newdefinition}
\subsection{Properties}
\label{sec:properties}
We now state some observations and then prove some properties of SUS-anchors.

\begin{observation}
Given a window \(W\), the smallest unique suffix is smaller than all longer
suffixes, because all longer suffixes must be unique and thus larger than the
smallest unique suffix.

\end{observation}

\begin{observation}
Removing the last character from the smallest unique substring results in
a non-unique substring.

\end{observation}

\begin{mytheorem}[SUS-anchors are forward]
SUS-anchors are forward for any character-based order.

\end{mytheorem}

\begin{proof}
Consider two consecutive windows \(W\) and \(W'\) with \(\sus(W) = W[i\dots j)\).
If \(i=0\), the scheme is trivially forward.
Otherwise, let \(0\leq i' < i-1\) be the start index of a suffix \(W'[i'\dots)\).
By the previous observation, \(W[i\dots) <_O W[i'+1\dots)\), and since \(W[i\dots)\)
is not a prefix of \(W[i'+1\dots)\) (for otherwise it would not be unique),
appending \(W'[\ell-1]\) to these two suffixes will not change their relative order.
Since \(W[i\dots)\) is unique in \(W\), \(W'[i-1\dots)\) is also unique in \(W'\),
and so \(W'[i-1\dots)\) is a smaller unique suffix than \(W'[i'\dots)\) for all \(i'<i-1\).
Thus, \(\sus(W')\) cannot start at an index less than \(i-1\), and
thus the SUS-anchor is forward.
\end{proof}

\begin{mytheorem}[Charged context]
Given two consecutive windows \(W\) and \(W'\) that together form a \emph{context}, the
sampled SUS-anchor changes if and only if either \(\sus(W)\) is a prefix of \(W\) or
\(\sus(W')\) is a suffix of \(W'\). In this case, the context is \emph{charged}.

\end{mytheorem}

\begin{proof}
Let \(\sus(W) = W[i\dots j)\) and \(\sus(W') = W'[i'\dots j')\).

If \(i=0\) is sampled, \(W'\) cannot sample the same position.
If \(W'[i'\dots j')=W[i'\dots \ell)\) is a suffix, then \(W'[i'\dots j'-1)\) is not
unique in \(W'\), and thus also not in \(W\). Thus, \(W\) must sample a different position.

For the other direction, assume that \(i' > i-1\), so
that different text positions are sampled.
If \(i=0\) or \(i'=\ell-1\) we are done, so
assume \(i>0\) and \(i'<\ell-1\).

If \(W[i\dots) < W[i'+1\dots)\), then for every possible character \(W'[\ell-1]\) that
can be appended, we have \(W'[i-1\dots) < W'[i'\dots)\), which is in contradiction
with the fact that \(W'[i'\dots)\) is smaller than all longer suffixes.
Thus, \(W[i'+1\dots)\) must be a prefix of \(W[i\dots)\), so that the suffix is not
unique. This means that \(W'[i'\dots \ell-1)\) is not unique in \(W'\), and thus, the
SUS \(W'[i'\dots j')\) can only end in \(j'=\ell\), as required.
\end{proof}
\subsection{Computing SUS-anchors}
\label{sec:computing-sus-anchors}
Via the equivalence to maximal suffixes, we can compute all lexicographic
SUS-anchors using the theoretical algorithm of
\cite{minimal-maximal-suffix-tcs} that requires \(O(n)\) space, \(O(n)\)
preprocessing, and supports \(O(1)\) queries for substring suffix-maximum.
When the alphabet has constant size, we can use the same trick as for bd-anchors
to reduce the space usage:
Run the construction algorithm on chunks of size \(2\ell\) that overlap by
\(\ell-1\), so that the total space usage is \(O(\ell)\) and the total time remains
\(O(n)\).

Below, we first adapt the method of \cite{minimal-maximal-suffix-tcs} for
computing the maximal suffix of each window to our streaming setting.
Then, we adapt it to compute SUS-anchors for any character-based order.

\begin{mytheorem}
For any character-based order \(O\), the SUS-anchors of a string of length \(n\)
over a constant-size alphabet can be computed in \(O(n)\) time and \(O(\ell)\)
space.

\end{mytheorem}

\parag{A streaming algorithm for the maximal suffix.}
We first adapt the algorithm of \cite{minimal-maximal-suffix-tcs} to our streaming setting.
Suppose that we have so far processed \(T[0\dots t)\).
The algorithm maintains a doubly linked list \(A=(a_0, a_1, a_2, \dots)\) of \emph{active} text positions \(a_i\)
such that \(T[a_i\dots t) > T[j\dots t)\) for all \(a_i < j < t\).
This list is \emph{decreasing} in the sense that \(T[a_0\dots t) > T[a_1\dots t) >
T[a_2\dots t) > \dots\), and it plays a similar role to the \emph{monotone queue} in
the classic minimizer algorithm \cite{simd-minimizers}.
Unlike that original case, here we also store a list of \emph{events} for every text
position that can modify the middle of the list.

As we shift the window and increment \(t\) to \(t'=t+1\), we remove \(a_0\) when \(a_0 < t'-\ell\).
After appending \(T[t]\), we get a new suffix \(T[t\dots
t')\). We update \(A\) before pushing the new element \(t\).
\begin{itemize}
\item \itemname{Repeatedly pop \(A\).} We remove the last element \(\last(A)=A_{|A|-1}\) from \(A\)
as long as \(T[\last(A)\dots t') < T[t\dots t')\).
\item \itemname{Add event.} Then, if \(A\) is not empty and \(T[\last(A)\dots) < T[t\dots)\)
(taking into account upcoming characters), compute \(\lambda=\lcp(T[t\dots),
  T[\last(A)\dots))\) and add the event ``check if we should drop the element
preceding \(t\) from \(A\)'' to
the event queue for position \(t+\lambda+1\).
\item \itemname{Push \(t\).} Then, push \(t\) to \(A\) and increment \(t\) to \(t' = t+1\).
\item \itemname{Handle event.} As soon as \(t''\) reaches position \(t+\lambda+1\), we check if
\(t\) is still in \(A\), and if so, if it has now become larger than its
preceding element. Then repeatedly remove the preceding element as long as \(T[t\dots t'')\) is
larger, and end when either there is no preceding element, or when \(T[a_i\dots t'') > T[t\dots t'')\). In that case, again check if the order will flip in the
future, and if so, add a new event to the event queue.
\end{itemize}

\parag{Practical considerations.}
To make this practical, all \(\lcp\) computations can be bounded to \(\ell\), and
the event queues can be stored in a ring buffer with just \(\ell+1\) slots for the
current and next \(\ell\) positions.
Originally, the \(\lcp\) computations are done in constant time using the suffix array.
Instead, we simply do a word-by-word scan on the
bit-representation of the text, which takes expected constant time on random
strings since LCPs of, say, over \(64\) bits are unlikely in random strings.

\enlargethispage{1em}
\parag{A more direct algorithm for SUS-anchors.}
We now make some modifications to use the algorithm above for SUS-anchors.
Most importantly, \(A\) will now be an \emph{increasing} list of suffixes, with the
smallest at the start, with the caveat that unlike usually a string is considered \emph{smaller}
than all its prefixes.
A final modification drops the event queues and allows us to replace the doubly
linked list \(A\) by a simple double-ended queue, as shown in \cref{alg}. 
For each \(a\in A\), we now additionally store the text position \(x_a\) where \(a\) becomes
``hot'', i.e., we store that \(a\) becomes the first element of \(A\) (and thus the
start of the maximal suffix of the window) when the
character \(T[x_a]\) enters the sliding window.
When comparing \(t\) to \(\last(A)\), there are three cases. 
\begin{enumerate}
\item \(A\) is empty, in which case we push \(t\) and set \(x_t = t\).
\item \(T[\last(A)\dots \last(A)+\ell) \leq_O T[t\dots t+\ell)\), in which case we push \(t\) and note that it
becomes hot (the smallest suffix) when \(\last(A)\) falls out of the window, i.e., \(x_t = \last(A)+\ell\).
\item \(T[\last(A)\dots \last(A)+\ell) >_O T[t\dots t+\ell)\): \(t\) ``takes over'' \(\last(A)\) when character
\(t+\lambda\) enters the window, where \(\lambda = \lcp(T[t\dots), T[\last(A)\dots))\).
\begin{itemize}
\item If \(\last(A)\) is already hot by that time (\(x_{\last(A)} < t+\lambda\)),
push \(t\) with \(x_t = \min(t + \lambda, a+\ell)\).
\item If not, \(\last(A)\) will never be hot. We pop it and (repeatedly) compare \(t\) to \(\last(A)\).
\end{itemize}
\end{enumerate}
Once this process finishes, the first element of \(A\) indicates the maximum suffix of the
current window. We then increment \(t\) to the next text position and check
whether we reached the position \(x_{a_1}\) where the next element of \(A\) becomes
hot (either because \(a_0\) falls out of the window or \(a_1\) introduces a new
maximal suffix). In that case, we pop \(a_0\).

\begin{algorithm}[t]
  \centering
  \caption{Pseudo code for the streaming algorithm to compute the SUS-anchors as
  given in \cref{sec:computing-sus-anchors}.}
  \label{alg}
  \begin{algorithmic}[1]
	\Procedure{SUS-anchors}{$T$, $\ell$, $O$}

    \State $A \gets [] $ \Comment{Deque of (position, hot-time) pairs.}
    \State $S \gets []$ \Comment{List of SUS-anchor positions of each window.}
    \For{$t \gets 0, 1, \dots, n-1$}
        \State $x_t \gets t$ \Comment{Moment where $T[t\dots)$ becomes hot.}
        \While{$A$ is not empty}
            \State $(a, x_a) \gets \last(A)$
            
            \If{$T[a\dots a{+}\ell) \leq_{O} T[t\dots t{+}\ell) $}
                \Comment{New suffix is larger.}
                \State $x_t \gets a + \ell$
                \Comment{Hot when $a$ exits the window.}
                \State \textbf{break}
            \Else \Comment{New suffix is smaller: takes over $a$ at time $x$}
                \State $\lambda \gets \lcp\bigl(T[t\dots t{+}\ell), T[a\dots a{+}\ell)\bigr)$
                \State $x \gets \min(t + \lambda, a + \ell)$
                \If{$x_a < x$}
                \Comment{$a$ becomes hot before $t$ takes over; push $t$ after $a$.}
                    \State $x_t \gets x$
                    \State \textbf{break}
                \Else \Comment{$a$ will never be hot; discard it and continue.}
                    \State $A.\mathsf{pop\_last}()$
                \EndIf
            \EndIf
        \EndWhile
        \State $A.\mathsf{push}((t, x_t))$
        \State $(p, x_p) \gets A.\mathsf{front}()$
        \State $S.\mathsf{push}(p)$ \Comment{SUS-anchor of $T[t-\ell+1\dots t+1)$}
        \If{$p + \ell - 1 = t \textbf{ or }\bigl(|A| > 1\textbf{ and }A.\textsf{index}(1).x \le t+1\bigr)$}
            \State $A.\mathsf{pop\_front()}$
            \Comment{Pop from $A$ if window ends or next element becomes hot.}
        \EndIf
    \EndFor
    \State \Return $S$
    \EndProcedure
	\end{algorithmic}
\end{algorithm}
\section{Results}
\label{sec:results}

\begin{figure}[p]
\centering
\makebox[\linewidth]{
\includesvg{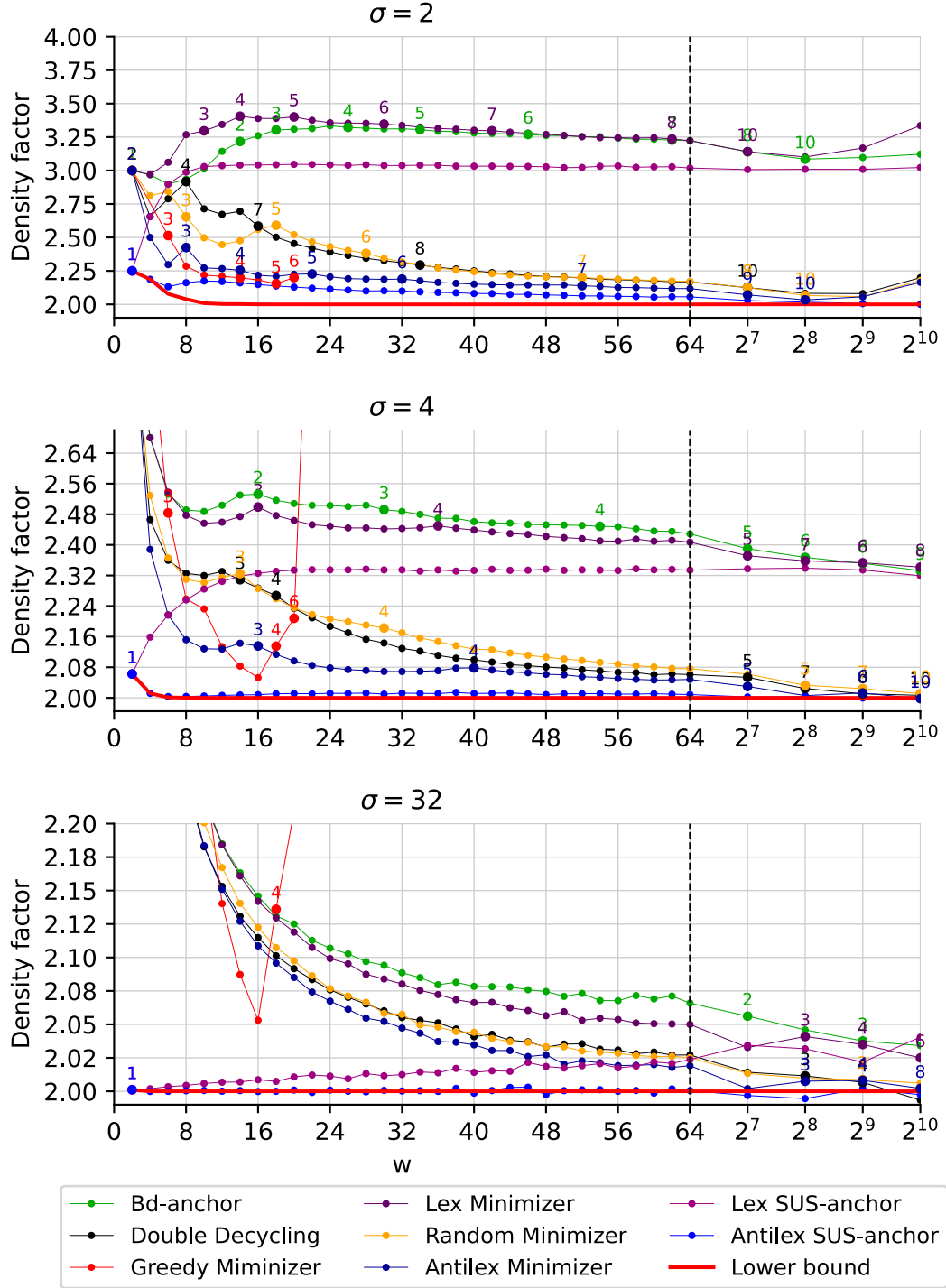}
}
\caption{\label{selection-s4}Comparison of the density factor (the density multiplied by \(w+1\)) of selection schemes for \(\sigma\in\{2,4,32\}\), and varying \(w\). The lower bound is shown in red. For each scheme (minimizers, bd-anchors \cite{bdanchors}, double decycling \cite{minimum-decycling-set}, greedy minimizer \cite{greedymini}), the best parameter \(r\) or \(k\) is chosen for each \(w\) (and \(w\) is decreased accordingly). Changes to the value of \(r\) or \(k\) are annotated. On the right of each plot, we show scaling for large \(w\). Each data point is the particular density on an independent random string of length \(10^7\). For large \(\sigma\) and \(w\), there remains some variance in the estimated density, occasionally leading to estimates below the lower bound. Sus-anchors are consistently better than minimizers, and the anti-lexicographic order is better than the random order which itself is better than the lexicographic order.}
\end{figure}

In \cref{selection-s4}, we compare the density of SUS-anchors against the bd-anchors
selection scheme as well as minimizer schemes. For the minimizer schemes, we choose the \(k\) that
minimizes the density, and we adjust \(w\) to \(w'=w+1-k\) accordingly to preserve
the overall window length \(\ell = w'+k-1\).
The true density is approximated by computing the \emph{particular density} on a
uniform random string of \(10^7\) characters over alphabets of size \(2\), \(4\), and \(32\).
A new independent random string is used for each data point.

Lexicographic minimizers are roughly on-par with bd-anchors, and lexicographic
SUS-anchors are slightly better but still far from optimal.
Double decycling \cite{minimum-decycling-set} and the random minimizer appear to converge to density factor
\(2\) as \(w\to\infty\) and have similar density. The greedy minimizer
\cite{greedymini} only performs well for small \(w\) and \(\sigma\), as it requires
precomputed orders that are not always available. (We did not implement
the logic to extrapolate from smaller orders.) The greedy minimizer is the
only scheme that improves over anti-lexicographic minimizers (excluding very
small \(w\leq 6\)).

The anti-lexicographic minimizer consistently improves over
the random minimizer and double decycling. The anti-lexicographic
SUS-anchors consistently have the best density of all, while also being parameter-free.
They get surprisingly close to the lower bound: they are away from the lower bound by less than $10\%$ for
$\sigma=2$ and less than $1\%$ for $\sigma=4$ and $\sigma=32$.
For $\sigma=32$, the anti-lexicographic sus-anchor is indistinguishable from
optimal for all $w\leq 64$, and improves over the random minimizer that is
$>1\%$ away from optimal here.

\begin{table}[t]
\centering
\caption{\label{table}Density factor and throughput of various selection schemes for \(w=31\) (left half) and \(w=63\) (right half) on a random string and on chromosome 19. For minimizer schemes, the value of \(k\) that minimizes the density on chromosome 19 is chosen. Schemes are sorted by decreasing density.}
\begin{tabular}{lrrrrrrrr}
  \toprule
  &
\multicolumn{3}{l}{$w=31$\quad\quad Density} & Thrpt. & 
  \multicolumn{3}{l}{$w=63$\quad \quad Density} & Thrpt. \\

\cmidrule(lr){2-5}\cmidrule(lr){6-9}
Scheme & \(k\) & Random & Ch.19 & Mbp/s & \(k\) & Random & Ch.19 & Mbp/s\\
\toprule
Bd-anchor & 2 & 2.59 & 3.04 & 2.0 & 3 & 2.53 & 3.36 & 1.0\\
Greedy minimizer & 4 & 3.62 & 2.75 & 44.9 & 5 & 3.68 & 2.86 & 43.9\\
Lex minimizer & 4 & 2.47 & 2.91 & 35.3 & 5 & 2.42 & 3.16 & 35.4\\
Lex SUS-anchor & - & 2.33 & 2.66 & 19.0 & - & 2.34 & 2.96 & 18.4\\
Random minimizer & 4 & 2.18 & 2.27 & 40.3 & 4 & 2.08 & 2.19 & 40.7\\
Double decycling & 5 & 2.18 & 2.23 & 13.6 & 5 & 2.10 & 2.13 & 13.7\\
Antilex minimizer & 3 & 2.07 & 2.29 & 36.2 & 4 & 2.05 & 2.42 & 35.0\\
Antilex SUS-anchor & - & 2.01 & 2.18 & 17.9 & - & 2.01 & 2.33 & 17.6\\
\bottomrule
\end{tabular}
\end{table}

\parag{Density and throughput on chromosome 19.}
In \cref{table}, we compare the density on a random string with the particular density
on chromosome 19 of CHM13 \cite{chm13v2}, which has length \(61.7\) Mbp.
Most schemes have slightly worse density on chromosome 19 than on real data.
The greedy minimizer has worse than expected performance in both
cases, which is likely due to the fact that our parameters are large and require
extrapolating from smaller cases.
The anti-lexicographic SUS-anchor has the best performance on random data, but
on chromosome 19, it is outperformed by the random minimizer and double decyling \cite{minimum-decycling-set}
for \(w=63\).

Additionally, \cref{table} shows the obtained
throughput of evaluating each scheme on an Intel i7-10750H CPU running at 2.6 GHz,
but we note that the implementations have not been optimized for speed.
In particular, for bd-anchors we use a very naive \(O(nw^2)\) implementation. Most
other methods are \(O(n)\) in practice, as can be seen from the fact that they do
not slow down from \(w=31\) to \(w=63\).
SIMD-minimizers \cite{simd-minimizers}, an optimized library for computing
random minimizers, has a throughput over \(500\) MB/s, which is more than
\(10\times\) faster than all of the methods shown here.

\parag{Intuition on SUS-anchor density.}
Low-density schemes sample positions that are as far apart as possible.
When sampling \(k\)-mers, this means that ideally each \(k\)-mer does not overlap
with the preceding or succeeding sampled \(k\)-mer, or else as little as possible.\footnote{In particular, a minimizer scheme that has density exactly equal to the lower
bound \cite{sampling-lower-bound} must never sample overlapping \(k\)-mers: 
The lower bound uses that at least two \(k\)-mers must be sampled on each
primitive cycle of \(w+k\) characters, and if two overlapping \(k\)-mers are sampled,
there must be at least a third sampled \(k\)-mer.
Similarly, SUS-anchors can only have optimal
density if they never overlap.}
Thus, shorter SUS-anchors are beneficial since they are both less likely to
overlap and can be the anchor for more consecutive \(k\)-mers.

Taken together, this means that short unique substrings should work well.
In particular, the anti-lexicographic order is better than the lexicographic
order since it avoids sampling overlapping substrings: in the lexicographic
order, consecutive substrings like \texttt{AAAB} and \texttt{AAB} are both small, causing
overlapping samples when a window starts with such a string.
The anti-lexicographic order prevents this by explicitly preferring a large
suffix after removing the first character.
A weakness of the ABB scheme is that it needs to sample
a relatively long substring before it becomes unique: in a window like
\texttt{ABDDCBA}, 6 of the 7 characters (\texttt{ABDDCB}) are part of the ABB-SUS, so that
this sample can only cover a few windows. The anti-lexicographic SUS is simply
\texttt{AB} instead.

This also explains the effect of the alphabet size: for \(\sigma=2\), on average a
substring needs to have length \(\log_2 w\) before it becomes unique. Thus, the
SUS-anchors are longer and more likely to overlap.
An example of a case where anti-lexicographic SUS-anchors still overlap is given
by the consecutive windows \texttt{ABABABB} and \texttt{BABABBB}, with anti-lexicographic
SUS-anchors \texttt{ABABA} and \texttt{ABA}. For alphabet size \(\sigma=2\), such cases are
somewhat common.
For \(\sigma=32\) on the other hand, most characters in the window are unique (for
\(w\leq 32\)) and most anchors will consist of only one or two characters (e.g.
just \texttt{A} or \texttt{AZ}) and thus rarely overlap. 
\section{Conclusion}
\label{sec:conclusion}
\enlargethispage{2em}
We introduced anti-lexicographic SUS-anchors, gave a linear-time algorithm to
compute them, and showed that they empirically get to within \(1\%\) of the density lower bound
for \(\sigma=4\). They are consistently the lowest density selection scheme of all
tested variants, while also being parameter-free. As a second best,
anti-lexicographic minimizers also perform better than all other schemes apart
from the greedy minimizer.

\parag{Future work.}
More work is needed to \emph{prove} that, for example, anti-lexicographic SUS anchors
have a close-to-optimal density of \(2/(w+1) + o(1/w)\).
Given how close to optimal these schemes already are, the fact that \emph{exactly}
optimal schemes are known to exist for small \(w\) \cite{sampling-lower-bound},
including manually constructed for \(\sigma=2\) and \(w\leq 12\), the question remains whether
``clean'', constant space schemes can be developed for generic \(w\). The \emph{spacers}
of \cite{10-minimizers} are similar to our own ongoing work, and provide a
first step in this direction.

On a practical level, further work is needed to improve the current
monotone-queue-based algorithm to a rescan-like approach \cite{simd-minimizers}
or even a branchless variant that could be implemented in parallel using SIMD.

It would also be interesting to use SUS-anchors in combination with the
mod-minimizer~\cite{modmini}.
In particular, the mod-minimizer requires a lower bound \(r = \Theta(\log_\sigma w)\) on the
length of the chosen minimizers (typically around \(4\) in practice), and it would be
interesting to see if ``mod-SUS-anchors'' can work without such a restriction.
More generally, SUS-anchors can be extended to \(k>1\) by simply excluding the last
\(k-1\) suffixes, and preliminary results show that this gives good density up to
\(k \approx \log_\sigma w\).

Lastly, the SUS-anchor is a forward scheme. It is known that non-forward schemes
can be (at least) slightly better at times and break the forward lower bound,
but this is not well understood.

\bibliographystyle{plainurl}
\bibliography{bibliography}

@Article{chm13v2,
  author       = {Rhie, Arang and Nurk, Sergey and Cechova, Monika and Hoyt,
                  Savannah J. and Taylor, Dylan J. and Altemose, Nicolas and
                  Hook, Paul W. and Koren, Sergey and Rautiainen, Mikko and
                  Alexandrov, Ivan A. and Allen, Jamie and Asri, Mobin and
                  Bzikadze, Andrey V. and Chen, Nae-Chyun and Chin, Chen-Shan
                  and Diekhans, Mark and Flicek, Paul and Formenti, Giulio and
                  Fungtammasan, Arkarachai and Garcia Giron, Carlos and
                  Garrison, Erik and Gershman, Ariel and Gerton, Jennifer L. and
                  Grady, Patrick G. S. and Guarracino, Andrea and Haggerty,
                  Leanne and Halabian, Reza and Hansen, Nancy F. and Harris,
                  Robert and Hartley, Gabrielle A. and Harvey, William T. and
                  Haukness, Marina and Heinz, Jakob and Hourlier, Thibaut and
                  Hubley, Robert M. and Hunt, Sarah E. and Hwang, Stephen and
                  Jain, Miten and Kesharwani, Rupesh K. and Lewis, Alexandra P.
                  and Li, Heng and Logsdon, Glennis A. and Lucas, Julian K. and
                  Makalowski, Wojciech and Markovic, Christopher and Martin,
                  Fergal J. and Mc Cartney, Ann M. and McCoy, Rajiv C. and
                  McDaniel, Jennifer and McNulty, Brandy M. and Medvedev, Paul
                  and Mikheenko, Alla and Munson, Katherine M. and Murphy,
                  Terence D. and Olsen, Hugh E. and Olson, Nathan D. and Paulin,
                  Luis F. and Porubsky, David and Potapova, Tamara and Ryabov,
                  Fedor and Salzberg, Steven L. and Sauria, Michael E. G. and
                  Sedlazeck, Fritz J. and Shafin, Kishwar and Shepelev, Valery
                  A. and Shumate, Alaina and Storer, Jessica M. and Surapaneni,
                  Likhitha and Taravella Oill, Angela M. and Thibaud-Nissen,
                  Françoise and Timp, Winston and Tomaszkiewicz, Marta and
                  Vollger, Mitchell R. and Walenz, Brian P. and Watwood, Allison
                  C. and Weissensteiner, Matthias H. and Wenger, Aaron M. and
                  Wilson, Melissa A. and Zarate, Samantha and Zhu, Yiming and
                  Zook, Justin M. and Eichler, Evan E. and O’Neill, Rachel J.
                  and Schatz, Michael C. and Miga, Karen H. and Makova, Kateryna
                  D. and Phillippy, Adam M.},
  title        = {The complete sequence of a human {Y} chromosome},
  journal      = {Nature},
  year         = 2023,
  volume       = 621,
  number       = 7978,
  month        = aug,
  pages        = {344–354},
  issn         = {1476-4687},
  doi          = {10.1038/s41586-023-06457-y},
  publisher    = {Springer Science and Business Media LLC}
}

@article{minimap,
  author =       {Li, Heng},
  title =        {Minimap and miniasm: fast mapping and de novo assembly for
                  noisy long sequences},
  journal =      {Bioinformatics},
  year =         2016,
  volume =       32,
  number =       14,
  month =        {March},
  pages =        {2103–2110},
  issn =         {1367-4803},
  doi =          {10.1093/bioinformatics/btw152},
  publisher =    {Oxford University Press (OUP)}
}

@inproceedings{winnowing,
  author =       {Schleimer, Saul and Wilkerson, Daniel S. and Aiken, Alex},
  title =        {Winnowing: local algorithms for document fingerprinting},
  year =         2003,
  booktitle =    {Proceedings of the 2003 ACM SIGMOD international conference on
                  Management of data},
  series =       {SIGMOD/PODS03},
  publisher =    {ACM},
  month =        jun,
  doi =          {10.1145/872757.872770},
  collection =   {SIGMOD/PODS03}
}

@article{sshash,
  author =       {Pibiri, Giulio Ermanno},
  title =        {Sparse and skew hashing of $k$-mers},
  journal =      {Bioinformatics},
  year =         2022,
  volume =       38,
  month =        jun,
  pages =        {i185–i194},
  issn =         {1367-4811},
  doi =          {10.1093/bioinformatics/btac245},
  publisher =    {Oxford University Press (OUP)}
}

@Article{sshash2,
  author       = {Pibiri, Giulio Ermanno and Patro, Rob},
  title        = {Optimizing sparse and skew hashing: faster $k$-mer
                  dictionaries},
  journal      = {Bioinformatics},
  year         = 2026,
  volume       = 42,
  month        = {July},
  issn         = {1367-4811},
  doi          = {10.1093/bioinformatics/btag264},
  publisher    = {Oxford University Press (OUP)}
}

@Article{kache-hash-preprint,
  author       = {Khan, Jamshed and Patro, Rob and Pandey, Prashant},
  title        = {$k$ache-hash: A dynamic, concurrent, and cache-efficient
                  hash table for streaming $k$-mer operations},
  year         = 2026,
  month        = Feb,
  doi          = {10.64898/2026.02.13.705625},
  journal      = {openRxiv}
}

@article{bdanchors,
  author =       {Loukides, Grigorios and Pissis, Solon P. and Sweering,
                  Michelle},
  title =        {Bidirectional String Anchors for Improved Text Indexing and
                  Top-$K$ Similarity Search},
  journal =      {IEEE Transactions on Knowledge and Data Engineering},
  year =         2023,
  volume =       35,
  number =       11,
  month =        nov,
  pages =        {11093–11111},
  issn =         {2326-3865},
  doi =          {10.1109/tkde.2022.3231780},
  publisher =    {Institute of Electrical and Electronics Engineers (IEEE)}
}

@article{minimizers,
  author =       {Roberts, Michael and Hayes, Wayne and Hunt, Brian R. and
                  Mount, Stephen M. and Yorke, James A.},
  title =        {Reducing storage requirements for biological sequence
                  comparison},
  journal =      {Bioinformatics},
  year =         2004,
  volume =       20,
  number =       18,
  month =        jul,
  pages =        {3363–3369},
  issn =         {1367-4803},
  doi =          {10.1093/bioinformatics/bth408},
  publisher =    {Oxford University Press (OUP)}
}

@article{improved-minimizers,
  author =       {Mar\c{c}ais, Guillaume and Pellow, David and Bork, Daniel and
                  Orenstein, Yaron and Shamir, Ron and Kingsford, Carl},
  title =        {Improving the performance of minimizers and winnowing schemes},
  journal =      {Bioinformatics},
  year =         2017,
  volume =       33,
  number =       14,
  month =        jul,
  pages =        {i110–i117},
  issn =         {1367-4811},
  doi =          {10.1093/bioinformatics/btx235},
  publisher =    {Oxford University Press (OUP)}
}

@article{mdbg,
  author =       {Ekim, Bar\i{}\c{s} and Berger, Bonnie and Chikhi, Rayan},
  title =        {Minimizer-space {De Bruijn} graphs: Whole-genome assembly of
                  long reads in minutes on a personal computer},
  journal =      {Cell Systems},
  year =         2021,
  volume =       12,
  number =       10,
  month =        oct,
  pages =        {958--968.e6},
  issn =         {2405-4712},
  doi =          {10.1016/j.cels.2021.08.009},
  publisher =    {Elsevier BV}
}

@inproceedings{docks-wabi,
  title =        {Compact Universal k-mer Hitting Sets},
  year =         2016,
  author =       {Orenstein, Yaron and Pellow, David and Mar\c{c}ais, Guillaume
                  and Shamir, Ron and Kingsford, Carl},
  booktitle =    {Algorithms in Bioinformatics},
  publisher =    {Springer International Publishing},
  isbn =         9783319436814,
  pages =        {257–268},
  doi =          {10.1007/978-3-319-43681-4_21},
  issn =         {1611-3349}
}

@article{docks,
  author =       {Orenstein, Yaron and Pellow, David and Mar\c{c}ais, Guillaume
                  and Shamir, Ron and Kingsford, Carl},
  title =        {Designing small universal k-mer hitting sets for improved
                  analysis of high-throughput sequencing},
  journal =      {PLOS Computational Biology},
  year =         2017,
  editor =       {Raphael, Benjamin J.},
  volume =       13,
  number =       10,
  month =        oct,
  pages =        {e1005777},
  issn =         {1553-7358},
  doi =          {10.1371/journal.pcbi.1005777},
  publisher =    {Public Library of Science (PLoS)}
}

@article{asymptotic-optimal-minimizers,
  author =       {Mar\c{c}ais, Guillaume and DeBlasio, Dan and Kingsford, Carl},
  title =        {Asymptotically optimal minimizers schemes},
  journal =      {Bioinformatics},
  year =         2018,
  volume =       34,
  number =       13,
  month =        jun,
  pages =        {i13–i22},
  issn =         {1367-4811},
  doi =          {10.1093/bioinformatics/bty258},
  publisher =    {Oxford University Press (OUP)}
}

@article{minimum-decycling-set,
  author =       {Pellow, David and Pu, Lianrong and Ekim, Bari\c{s} and Kotlar,
                  Lior and Berger, Bonnie and Shamir, Ron and Orenstein, Yaron},
  title =        {Efficient minimizer orders for large values of $k$ using
                  minimum decycling sets},
  journal =      {Genome Research},
  year =         2023,
  month =        aug,
  issn =         {1549-5469},
  doi =          {10.1101/gr.277644.123},
  publisher =    {Cold Spring Harbor Laboratory}
}

@article{small-uhs,
  author =       {Zheng, Hongyu and Kingsford, Carl and Mar\c{c}ais, Guillaume},
  title =        {Lower Density Selection Schemes via Small Universal Hitting
                  Sets with Short Remaining Path Length},
  journal =      {Journal of Computational Biology},
  year =         2021,
  volume =       28,
  number =       4,
  month =        apr,
  pages =        {395–409},
  issn =         {1557-8666},
  doi =          {10.1089/cmb.2020.0432},
  publisher =    {Mary Ann Liebert Inc}
}

@inproceedings{modmini,
  author =       {Groot Koerkamp, Ragnar and Pibiri, Giulio Ermanno},
  title =        {The mod-minimizer: A Simple and Efficient Sampling Algorithm
                  for Long $k$-mers},
  booktitle =    {WABI 2024},
  pages =        {11:1--11:23},
  series =       {LIPIcs},
  isbn =         {978-3-95977-340-9},
  issn =         {1868-8969},
  year =         2024,
  volume =       312,
  urn =          {urn:nbn:de:0030-drops-206552},
  doi =          {10.4230/LIPIcs.WABI.2024.11}
}

@article{syncmers,
  author =       {Edgar, Robert},
  title =        {Syncmers are more sensitive than minimizers for selecting
                  conserved k‑mers in biological sequences},
  journal =      {PeerJ},
  year =         2021,
  volume =       9,
  month =        feb,
  pages =        {e10805},
  issn =         {2167-8359},
  doi =          {10.7717/peerj.10805},
  publisher =    {PeerJ}
}

@Article{sampling-lower-bound,
  author =       {Kille, Bryce and Groot Koerkamp, Ragnar and McAdams, Drake and
                  Liu, Alan and Treangen, Todd J},
  title =        {A {near-tight} lower bound on the density of forward sampling
                  schemes},
  journal =      {Bioinformatics},
  year =         2024,
  editor =       {Ponty, Yann},
  month =        dec,
  issn =         {1367-4811},
  doi =          {10.1093/bioinformatics/btae736},
  publisher =    {Oxford University Press (OUP)}
}

@InProceedings{random-mini-density,
  title =        {Expected Density of Random Minimizers},
  year =         2025,
  author =       {Golan, Shay and Shur, Arseny M.},
  booktitle =    {SOFSEM 2025: Theory and Practice of Computer Science},
  publisher =    {Springer Nature Switzerland},
  isbn =         9783031826702,
  pages =        {347–360},
  doi =          {10.1007/978-3-031-82670-2_25},
  ISSN =         {1611-3349}
}

@Article{greedymini,
  author =       {Golan, Shay and Tziony, Ido and Kraus, Matan and Orenstein,
                  Yaron and Shur, Arseny},
  title =        {{GreedyMini: generating low-density DNA minimizers}},
  journal =      {Bioinformatics},
  year =         2025,
  volume =       41,
  month =        jul,
  pages =        {275-284},
  issn =         {1367-4811},
  doi =          {10.1093/bioinformatics/btaf251},
  publisher =    {Oxford University Press (OUP)}
}

@Article{oc-modmini,
  author =       {Groot Koerkamp, Ragnar and Liu, Daniel and Pibiri, Giulio
                  Ermanno},
  title =        {The open-closed mod-minimizer algorithm},
  journal =      {Algorithms for Molecular Biology},
  year =         2025,
  volume =       20,
  number =       4,
  month =        {March},
  issn =         {1748-7188},
  doi =          {10.1186/s13015-025-00270-0},
  publisher =    {Springer Science and Business Media LLC}
}

@article{minimizer-sketches,
  title =        {Creating and using minimizer sketches in computational
                  genomics},
  author =       {Zheng, Hongyu and Mar{\c{c}}ais, Guillaume and Kingsford,
                  Carl},
  doi =          {10.1089/cmb.2023.0094},
  journal =      {Journal of Computational Biology},
  number =       12,
  pages =        {1251--1276},
  publisher =    {Mary Ann Liebert, Inc., publishers 140 Huguenot Street, 3rd
                  Floor New~…},
  volume =       30,
  year =         2023,
}

@Article{minimally-overlapping-words,
  author =       {Frith, Martin C and Noé, Laurent and Kucherov, Gregory},
  title =        {Minimally overlapping words for sequence similarity search},
  journal =      {Bioinformatics},
  year =         2020,
  editor =       {Ponty, Yann},
  volume =       36,
  number =       {22–23},
  month =        dec,
  pages =        {5344–5350},
  issn =         {1367-4811},
  doi =          {10.1093/bioinformatics/btaa1054},
  publisher =    {Oxford University Press (OUP)}
}

@Article{non-overlapping-codes,
  author =       {Blackburn, Simon R.},
  title =        {Non-Overlapping Codes},
  journal =      {IEEE Transactions on Information Theory},
  year =         2015,
  volume =       61,
  number =       9,
  month =        sep,
  pages =        {4890–4894},
  issn =         {1557-9654},
  doi =          {10.1109/tit.2015.2456634},
  publisher =    {Institute of Electrical and Electronics Engineers (IEEE)}
}

@InProceedings{simd-minimizers,
  author =       {Groot Koerkamp, Ragnar and Martayan, Igor},
  title =        {{SimdMinimizers: Computing Random Minimizers, fast}},
  booktitle =    {{SEA 2025}},
  pages =        {20:1--20:19},
  series =       {LIPIcs},
  ISBN =         {978-3-95977-375-1},
  ISSN =         {1868-8969},
  year =         2025,
  volume =       338,
  URN =          {urn:nbn:de:0030-drops-232581},
  doi =          {10.4230/LIPIcs.SEA.2025.20}
}

@InProceedings{u-index,
  author =       {Ayad, Lorraine A. K. and Fici, Gabriele and Groot Koerkamp,
                  Ragnar and Loukides, Grigorios and Patro, Rob and Pibiri,
                  Giulio Ermanno and Pissis, Solon P.},
  title =        {{U-Index: A Universal Indexing Framework for Matching Long
                  Patterns}},
  booktitle =    {{SEA 2025}},
  pages =        {4:1--4:18},
  series =       {LIPIcs},
  ISBN =         {978-3-95977-375-1},
  ISSN =         {1868-8969},
  year =         2025,
  volume =       338,
  URN =          {urn:nbn:de:0030-drops-232420},
  doi =          {10.4230/LIPIcs.SEA.2025.4}
}

@Article{mashmap,
  author =       {Jain, Chirag and Dilthey, Alexander and Koren, Sergey and
                  Aluru, Srinivas and Phillippy, Adam M.},
  title =        {A Fast Approximate Algorithm for Mapping Long Reads to Large
                  Reference Databases},
  journal =      {Journal of Computational Biology},
  year =         2018,
  volume =       25,
  number =       7,
  month =        jul,
  pages =        {766–779},
  issn =         {1557-8666},
  doi =          {10.1089/cmb.2018.0036},
  publisher =    {Mary Ann Liebert Inc}
}

@Article{selection,
  author =       {Blum, Manuel and Floyd, Robert W. and Pratt, Vaughan and
                  Rivest, Ronald L. and Tarjan, Robert E.},
  title =        {Time bounds for selection},
  journal =      {Journal of Computer and System Sciences},
  year =         1973,
  volume =       7,
  number =       4,
  month =        aug,
  pages =        {448–461},
  issn =         {0022-0000},
  doi =          {10.1016/s0022-0000(73)80033-9},
  publisher =    {Elsevier BV}
}

@phdthesis{thesis,
  author =       {Groot Koerkamp, Ragnar},
  title =        {Optimal Throughput Bioinformatics},
  school =       {ETH Zurich},
  year =         2025,
  doi =          {10.3929/ETHZ-C-000783091},
  copyright =    {http://rightsstatements.org/page/InC-NC/1.0/}
}

@Article{deacon-preprint,
  author       = {Constantinides, Bede and Lees, John and Crook, Derrick W},
  title        = {Deacon: fast sequence filtering and contaminant depletion},
  year         = 2025,
  month        = jun,
  doi          = {10.1101/2025.06.09.658732},
  journal      = {bioRxiv}
}

@InProceedings{minimal-maximal-suffix-cpm13,
  title        = {On Minimal and Maximal Suffixes of a Substring},
  year         = 2013,
  author       = {Babenko, Maxim and Kolesnichenko, Ignat and Starikovskaya, Tatiana},
  booktitle    = {Combinatorial Pattern Matching},
  publisher    = {Springer Berlin Heidelberg},
  isbn         = 9783642389054,
  pages        = {28–37},
  doi          = {10.1007/978-3-642-38905-4_5},
  ISSN         = {1611-3349}
}

@article{minimal-maximal-suffix-tcs,
title = {Computing minimal and maximal suffixes of a substring},
journal = {Theoretical Computer Science},
volume = {638},
pages = {112-121},
year = {2016},
note = {Pattern Matching, Text Data Structures and Compression},
issn = {0304-3975},
doi = {10.1016/j.tcs.2015.08.023},
author = {Maxim Babenko and Paweł Gawrychowski and Tomasz Kociumaka and Ignat Kolesnichenko and Tatiana Starikovskaya}
}

@InProceedings{minimal-suffix-cpm16,
  author =	{Kociumaka, Tomasz},
  title =	{{Minimal Suffix and Rotation of a Substring in Optimal Time}},
  booktitle =	{27th Annual Symposium on Combinatorial Pattern Matching (CPM 2016)},
  pages =	{28:1--28:12},
  series =	{Leibniz International Proceedings in Informatics (LIPIcs)},
  ISBN =	{978-3-95977-012-5},
  ISSN =	{1868-8969},
  year =	{2016},
  volume =	{54},
  editor =	{Grossi, Roberto and Lewenstein, Moshe},
  publisher =	{Schloss Dagstuhl -- Leibniz-Zentrum f{\"u}r Informatik},
  address =	{Dagstuhl, Germany},
  URN =		{urn:nbn:de:0030-drops-60626},
  doi =		{10.4230/LIPIcs.CPM.2016.28}
}

@article{factorizing-words,
title = {Factorizing words over an ordered alphabet},
journal = {Journal of Algorithms},
volume = {4},
number = {4},
pages = {363-381},
year = {1983},
issn = {0196-6774},
doi = {10.1016/0196-6774(83)90017-2},
author = {Jean-Pierre Duval}
}

@article{smallest-rotation,
title = {Lexicographically least circular substrings},
journal = {Information Processing Letters},
volume = {10},
number = {4},
pages = {240-242},
year = {1980},
issn = {0020-0190},
doi = {https://doi.org/10.1016/0020-0190(80)90149-0},
author = {Kellogg S. Booth}
}

@article {10-minimizers,
	author = {Shur, Arseny and Tziony, Ido and Orenstein, Yaron},
	title = {10-minimizers: a promising class of constant-space minimizers},
	elocation-id = {2026.03.16.712052},
	year = {2026},
	doi = {10.64898/2026.03.16.712052},
	publisher = {Cold Spring Harbor Laboratory},
	journal = {bioRxiv}
}

@article{vigemers,
  author       = {Ingels, Florian and Limasset, Antoine and Marchet, Camille and
                  Salson, Mikaël},
  title        = {Vigemers: on the number of $k$-mers sharing the same XOR-based
                  minimizer},
  year         = 2026,
  doi          = {10.48550/arXiv.2602.03337},
  journal      = {arXiv},
  copyright    = {Creative Commons Attribution 4.0 International}
}

@inProceedings{multiminimizers,
  author       = {Ingels, Florian and Robidou, Lucas and Martayan, Igor and
                  Marchet, Camille and Limasset, Antoine},
  title        = {Minimizer Density revisited: Models and Multiminimizers},
  year         = 2026,
  month        = Nov,
  doi          = {10.1101/2025.11.21.689688},
  booktitle    = {RECOMB},
}

@article{on-minimizers-of-minimum-density,
  author       = {Shur, Arseny},
  title        = {On Minimizers of Minimum Density},
  year         = 2025,
  doi          = {10.48550/arXiv.2506.05277},
  journal    = {arXiv},
  copyright    = {Creative Commons Attribution 4.0 International}
}

@Article{optmini,
  author       = {Shur, Arseny and Tziony, Ido and Orenstein, Yaron},
  title        = {Generating minimum-density minimizers},
  year         = 2026,
  month        = Jan,
  doi          = {10.64898/2026.01.25.701585},
  journal      = {openRxiv}
}

@InProceedings{sus-anchor,
  author =	{Groot Koerkamp, Ragnar},
  title =	{{The Anti-Lexicographic SUS-Anchor: An Empirically Optimal Selection Scheme}},
  booktitle =	{26th International Conference on Algorithms for Bioinformatics (WABI 2026)},
  pages =	{22:1--22:15},
  series =	{Leibniz International Proceedings in Informatics (LIPIcs)},
  ISBN =	{978-3-95977-446-8},
  ISSN =	{1868-8969},
  year =	{2026},
  volume =	{390},
  editor =	{El-Mabrouk, Nadia and Vandin, Fabio},
  publisher =	{Schloss Dagstuhl -- Leibniz-Zentrum f{\"u}r Informatik},
  address =	{Dagstuhl, Germany},
  URN =		{urn:nbn:de:0030-drops-275261},
  doi =		{10.4230/LIPIcs.WABI.2026.22}
}

@Article{lc-anchors,
  author       = {Ayad, Lorraine A. K. and Loukides, Grigorios and Pissis, Solon P.},
  title        = {Text indexing for long patterns using locally consistent anchors},
  journal      = {The VLDB Journal},
  year         = 2025,
  volume       = 34,
  number       = 5,
  month        = {July},
  issn         = {0949-877X},
  doi          = {10.1007/s00778-025-00935-7},
  publisher    = {Springer Science and Business Media LLC}
}

\end{document}